\documentclass[12pt,reqno,a4paper]{article}
\usepackage{amsmath,amssymb,dsfont,graphicx,latexsym,
epsf,cancel,epic,eepic,amsthm,tikz-cd,mathtools}
\usepackage[colorlinks=false]{hyperref}
\usepackage{mathpazo}
\usepackage[utf8]{inputenc}
\usepackage{tikz-cd}

\newtheorem{theorem}{Theorem}
\newtheorem{proposition}{Proposition}

\newtheorem{definition}{Definition}

\newcommand{\bbZ}{{\mathbb Z}}

\newcommand{\bbC}{{\mathbb C}}
\newcommand{\bbP}{{\mathbb P}}

\newcommand{\cC}{{\mathcal C}}
\newcommand{\cE}{{\mathcal E}}

\newcommand{\cQ}{{\mathcal Q}}

\begin{document}

\title{Geometry of autonomous discrete Painlev\'e equations related to the Weyl group $W(E_8^{(1)})$}
\author{Jaume Alonso \and Yuri B.\ Suris}
\date{}
\maketitle

\vspace{-1cm}

\begin{center}
Institut f\"ur Mathematik, MA 7-1\\
Technische Universit\"at Berlin \\
Str. des 17. Juni 136,
10623 Berlin, Germany\\
E-mail: {\tt  alonso@math.tu-berlin.de, suris@math.tu-berlin.de}
\end{center}

\begin{abstract}

Discrete Painlev\'e equations are integrable two-dimensional birational maps associated to a family of generalized Halphen surfaces. The latter can be seen either as $\bbP^2$ blown up at nine points or as $\bbP^1\times\bbP^1$ blown up at eight points. These maps become autonomous if the blow-up points are in a special position (support a pencil of cubic curves in $\bbP^2$, respectively a pencil of biquadratic curves in $\bbP^1\times\bbP^1$), so that the generalized Halphen surfaces become rational elliptic surfaces. In the generic case, the symmetry of a discrete Painlev\'e equation is the Weyl group $W(E_8^{(1)})$. One has a system of commuting maps which correspond to translational elements of $W(E_8^{(1)})$ associated to the roots of the lattice $E_8^{(1)}$. In the present note, we give a geometric construction of these commuting maps. For this, we use some novel birational involutions based on the above mentioned pencils of curves.
\end{abstract}

\section{Introduction}
\label{intro}

QRT maps are celebrated integrable 2D maps introduced in \cite{QRT1, QRT2} and much studied since then, see also \cite{QRT book}. They can be simply defined using a \emph{pencil of biquadratic curves}
$$
\cQ_\mu=\Big\{(x,y)\in \bbC^2: Q_\mu(x,y):=Q_0(x,y)-\mu Q_\infty(x,y)=0\Big\},
$$
where $Q_0,Q_\infty$ are two polynomials of bidegree (2,2). Actually, we can consider this pencil in a compactification $\bbP^1\times \bbP^1$ of $\bbC^2$. The \emph{base set} $\mathcal B$ of the pencil is defined as the set of points through which all curves of the pencil pass or, equivalently, as the intersection $\{Q_0(x,y)=0\}\cap\{Q_\infty(x,y)=0\}$. It consists of eight points, counted with multiplicity. 

One defines the \emph{vertical switch} $i_1$ and the \emph{horizontal switch} $i_2$ as follows. Through any point $(x_0,y_0)\not\in \mathcal B$, there passes exactly one curve of the pencil, defined by the parameter value $\mu=\mu(x_0,y_0)=Q_0(x_0,y_0)/Q_\infty(x_0,y_0)$. The vertical line $\{x=x_0\}$ intersects $\cQ_\mu$ at exactly one further point $(x_0,y_1)$ which is defined to be $i_1(x_0,y_0)$; similarly, the horizontal line $\{y=y_0\}$ intersects $\cQ_\mu$ at exactly one further point $(x_1,y_0)$ which is defined to be $i_2(x_0,y_0)$. The QRT map is defined as
$$
f=i_1\circ i_2.
$$
Each of the maps $i_1$, $i_2$ is a birational involution on $\bbP^1\times \bbP^1$ with indeterminacy set $\mathcal B$. Likewise, the QRT map $f$ is a (dynamically nontrivial) birational map on $\bbP^1\times \bbP^1$, having $\mu(x,y)=Q_0(x,y)/Q_\infty(x,y)$ as an integral of motion. A generic fiber $\cQ_\mu$ is an elliptic curve, and $f$ acts on it as a translation with respect to the corresponding addition law. 

Already at the early days of the theory it was realised that this construction is morally (and birationally) equivalent to another one, producing integrable birational maps on $\bbP^2$ preserving a pencil of cubic curves. A quick definition is as follows.  Consider a \emph{pencil of cubic curves}
$$
\cC_\mu=\Big\{(x,y)\in \bbC^2: C_\mu(x,y):=C_0(x,y)-\mu C_\infty(x,y)=0\Big\},
$$
where $C_0,C_\infty$ are two polynomials of degree 3. Upon introducing homogeneous coordinates $[X:Y:Z]=[x:y:1]$ and lifting the polynomials $C_0,C_\infty$ to homogeneous polynomials of $X,Y,Z$ of degree 3, we can consider the situation in the compactification $\bbP^2$ of $\bbC^2$. The \emph{base set} $\mathcal B$ of the pencil is defined as the set of points through which all curves of the pencil pass or, equivalently, as the intersection of the two base cubics of the pencil. It consists of nine points $P_i$, $i=1,\ldots,9$, counted with multiplicity. 

For $i=1,\ldots,9$, one defines the \emph{Manin involution} $I_{i}$ as follows. For any point $P\not\in \mathcal B$, we define $I_i(P)$ as the third intersection point of the line $(P_iP)$ with the unique curve $C_\mu$ of the pencil passing through $P$. Now one can define \emph{Manin maps} $M_{ij}=I_{i}\circ I_{j}$ for any two distinct $i,j=1,\ldots,9$. Each $M_{ij}$ is a birational map on $\bbP^2$ preserving the pencil of cubic curves, which acts on a generic fiber as a translation with respect to the corresponding addition law. Manin involutions and Manin maps have been described and applied to algebro-geometric questions in \cite{Manin, Dolgachev}. In the theory of integrable systems, their relation to QRT maps was pointed out by Veselov \cite[p. 35]{Veselov}, see also \cite{Tsuda}. In Veselov's formulation, one sends the base points $P_i$, $P_j$ to $[1:0:0]$, $[0:1:0]$ by a projective transformation, then in the affine coordinates $[X:Y:Z]=[x:y:1]$ the cubic curves become biquadratic, and the Manin involutions $I_i$, $I_j$ become the QRT switches. This argument also delivers several birational maps commuting with a given QRT map, namely all other $M_{\ell m}$.

Manin involutions and Manin maps made a remarkable appearance in the geometric interpretation of the elliptic Painlev\'e equation \cite{KMNOY}. Recall that, according to the modern understanding of the discrete Painlev\'e equations, gradually established upon appearing of the pioneering paper by Sakai \cite{Sakai}, any such equation is a birational map representing a translation contained in the affine Weyl group $W(E_8^{(1)})$ acting on configurations of ten points in $\bbP^2$. The first nine points play the role of parameters, while the tenth one plays the role of the dependent variable of a discrete Painlev\'e equation. In general position, the first nine points undergo a certain evolution, rendering the non-autonomous nature of the discrete Painlev\'e equation. However, if these nine points are in a special position, supporting a pencil of cubic curves, they do not evolve, and the discrete Painlev\'e equation becomes autonomous. This is the case we are mainly interested in this paper.

Thus, autonomous versions of discrete Painlev\'e equations represent translations contained in the affine Weyl group $W(E_8^{(1)})$. Translations constitute a commutative subgroup of $W(E_8^{(1)})$ generated by those corresponding to the roots of the system $E_8^{(1)}$. The Weyl group acts by birational transformations (generated by Cremona inversions) on $\bbP^2$, or, better, on the blow-up variety
\begin{equation}\label{X}
X=Bl_{P_1,\ldots,P_9}\bbP^2.
\end{equation}
This induces an action by linear transfromations (generated by reflections) on the Picard lattice 
\begin{equation}\label{Pic X}
{\rm Pic}(X)=\bbZ\cE_0\oplus\bbZ\cE_1\oplus\ldots\oplus\bbZ\cE_9.
\end{equation}
Here $\cE_0$ is the divisor class of the total transform of a generic line, while $\cE_i$ are exceptional divisors of the blow-up of $\bbP^2$ at $P_i$. The Picard lattice carries the Minkowski scalar product (intersection number), with the only non-vanishing products of the basis elements being
\begin{equation}\label{scal prod X}
\langle \cE_0,\cE_0\rangle=1,\quad \langle \cE_i,\cE_i\rangle =-1, \; i=1,\ldots, 9.
\end{equation}
The elements $\alpha$ of the root system $E_8^{(1)}$ are defined by the relations $\langle\alpha,\delta\rangle=0$ and $\langle\alpha,\alpha\rangle=-2$, where
\begin{equation}\label{delta X}
\delta=3\cE_0-\cE_1-\ldots-\cE_9
\end{equation}
is the null root (the anti-canonical divisor). Any such $\alpha$ is equivalent $\!\!\!\pmod{\bbZ\delta}$ to one of the 120 roots given by
\begin{equation}
\cE_i-\cE_j \;\; (i<j)  \quad {\rm and} \quad \cE_0-\cE_i-\cE_j-\cE_k \;\; (i< j<k)
\end{equation}
(${9\choose 2}=36$ of the first kind and ${9\choose 3}=84$ of the second kind), or to one opposite to those. The action of the translation $T_\alpha$ corresponding to a root $\alpha$ on ${\rm Pic}(X)$ is given by the Kac formula \cite{Kac}:
\begin{equation}\label{Kac formula}
T_\alpha(\lambda)=\lambda+\langle\delta,\lambda\rangle\alpha-\Big(\frac{1}{2}\langle\delta,\lambda\rangle\langle\alpha,\alpha\rangle+\langle\alpha,\lambda\rangle\Big)\delta.
\end{equation}

The birational action of the Weyl group on $\bbP^2$ is defined in terms of the Cremona inversions centered at the triples of the base points. However, we are interested in a geometric description. Such a description for the action of $T_{\cE_i-\cE_j}$ was given in \cite{KMNOY}:
\begin{proposition}\label{prop Manin}{\bf \cite{KMNOY}}
The birational map $I_{j}\circ I_{i}$ on $\bbP^2$  induces the action of $T_{\cE_i-\cE_j}$ on ${\rm Pic}(X)$.
\end{proposition}
Since the proof of this statement can serve as a blueprint for all other proofs of this paper, and for convenience of the reader, we reproduce it here.

\smallskip
\noindent
{\it Proof of Proposition \ref{prop Manin}.}  A direct computation with the Kac formula \eqref{Kac formula} gives the action of $T_{\cE_i-\cE_j}$ on ${\rm Pic}(X)$:
\begin{eqnarray}
T_{\cE_i-\cE_j}(\cE_0) & = & \cE_0+3(\cE_i-\cE_j)+3\delta \;=\; 10\cE_0-6\cE_j-3\sum_{k\neq i,j}\cE_k,\label{Tij 0}\\
T_{\cE_i-\cE_j}(\cE_i) & = & \cE_i+(\cE_i-\cE_j)+2\delta \;=\; 6\cE_0-3\cE_j-2\sum_{k\neq i,j}\cE_k, \label{Tij i}\\
T_{\cE_i-\cE_j}(\cE_j) & = & \cE_j+(\cE_i-\cE_j) \;=\; \cE_i, \label{Tij j}\\
T_{\cE_i-\cE_j}(\cE_\ell) & = & \cE_\ell+(\cE_i-\cE_j) +\delta\;=\; 3\cE_0-2\cE_j-\sum_{k\neq i,j,\ell}\cE_k. \label{Tij l}
\end{eqnarray}
On the other hand, the action of $I_i$ on ${\rm Pic}(X)$ was already computed in \cite{Manin}:
\begin{eqnarray}
(I_{i})_*(\cE_0) & = & 5\cE_0-4\cE_i-\sum_{k\neq i} \cE_k, \label{Ii 0}\\
(I_{i})_*(\cE_i) & = & 4\cE_0-3\cE_i-\sum_{k\neq i}\cE_k, \label{Ii i}\\
(I_{i})_*(\cE_j) & = & \cE_0 -\cE_i-\cE_j, \;\; j\neq i. \label{Ii j}
\end{eqnarray}
Now a straightforward computation of the action of $(I_j\circ I_i)_*$ on ${\rm Pic}(X)$ shows that it coincides with \eqref{Tij 0}--\eqref{Tij l}.
\qed

\medskip
Finding a similar description for the action of $T_{\cE_0-\cE_i-\cE_j-\cE_k}$ was mentioned in \cite{KMNOY} as an open problem, and seems to remain open to this day. Our main result here is the answer to this quest.

\section{Involutions defined by pencils of cubics in $\mathbb{P}^2$}
\label{s:cubics}

\begin{definition}\label{def cubic inv 1234}
Consider a pencil of cubics in $\bbP^2$ with base points $P_i$, $i=1,\ldots,9$. For any four distinct base points $P_i,P_j,P_k,P_\ell$, we define the involution $I_{i,j,k,\ell}$ as follows. Through a generic point $P$ passes a unique cubic $C_P$ of the pencil. Moreover, through $P_i,P_j,P_k,P_\ell,P$ passes a unique conic $S$. We set  
$$
I_{i,j,k,\ell}(P) = \widetilde{P},
$$
where $\widetilde{P}$ is the sixth intersection point of $C_P$ with $S$, so that $C_P \cap S = \{P_i, P_j,P_k,P_\ell,P,\widetilde{P}\}$.
\end{definition}

With the help of these novel geometric involutions, the solution of the above mentioned problem is formulated as follows.

\begin{theorem}\label{th cubic I1234 circ I4}
The birational map $I_{i,j,k,\ell}\circ I_{\ell}$ on $\bbP^2$  induces the action of $T_{\cE_0-\cE_i-\cE_j-\cE_k}$ on ${\rm Pic}(X)$.
\end{theorem}
\begin{proof} For notational convenience, we take $(i,j,k,\ell)=(1,2,3,4)$. We compare the action of the both sides on the Picard lattice. First of all, a direct computation with the Kac formula \eqref{Kac formula} gives:
\begin{eqnarray*}
T_{\cE_0-\cE_1-\cE_2-\cE_3}(\cE_0) & = & 10\cE_0-5(\cE_1+\cE_2+\cE_3)-2(\cE_4+\ldots+\cE_9),\\
T_{\cE_0-\cE_1-\cE_2-\cE_3}(\cE_i) & = & \cE_0-\cE_1-\cE_2-\cE_3+\cE_i,\;\;i=1,2,3,\\
T_{\cE_0-\cE_1-\cE_2-\cE_3}(\cE_i) & = & 4\cE_0-2(\cE_1+\cE_2+\cE_3)-(\cE_4+\ldots+\cE_9)+\cE_i,\;\;i=4,\ldots,9.
\end{eqnarray*}
Next, we compute the action of $I_{1,2,3,4}$ on ${\rm Pic}(X)$.
\begin{proposition}\label{prop I1234}
The action of $I_{1,2,3,4}$ on ${\rm Pic}(X)$ is given by:
\begin{eqnarray}
(I_{1,2,3,4})_*(\cE_0) & = &11\cE_0-5(\cE_1+\ldots+\cE_4)-2(\cE_5+\ldots+\cE_9), \label{I1234 0}\\
(I_{1,2,3,4})_*(\cE_i) & = & 5\cE_0-2(\cE_1+\ldots+\cE_4)-\cE_i-(\cE_5+\ldots+\cE_9),\;\;i=1,\ldots,4,\label{I1234 1..4}\\
(I_{1,2,3,4})_*(\cE_i) & = & 2\cE_0-(\cE_1+\ldots+\cE_4)-\cE_i,\;\;i=5,\ldots,9.\label{I1234 5..9}
\end{eqnarray}
\end{proposition}
\begin{proof}
We observe that a generic cubic curve of the pencil is mapped to itself, therefore
$$
(I_{1,2,3,4})_*(3\cE_0-\cE_1-\ldots-\cE_9)=3\cE_0-\cE_1-\ldots-\cE_9.
$$ 
Any conic through $P_1,P_2,P_3,P_4$ is mapped to itself, therefore
$$
(I_{1,2,3,4})_*(2\cE_0-\cE_1-\ldots-\cE_4)=2\cE_0-\cE_1-\ldots-\cE_4.
$$ 
The latter equation is actually a consequence of stronger relations. By definition, the point $\widetilde P=I_{1,2,3,4}(P)$ belongs to the conic through $P_1,P_2,P_3,P_4$ and $P$. In particular, if $P$ belongs to the line $(P_1P_2)$ then $\widetilde P$ belongs to the line $(P_3P_4)$, and vice versa, so that 
$$
(I_{1,2,3,4})_*(\cE_0-\cE_1-\cE_2)=\cE_0-\cE_3-\cE_4,  \quad (I_{1,2,3,4})_*(\cE_0-\cE_3-\cE_4)=\cE_0-\cE_1-\cE_2.
$$
Similarly, we find two further pairs of relations:
$$
(I_{1,2,3,4})_*(\cE_0-\cE_1-\cE_3)=\cE_0-\cE_2-\cE_4, \quad (I_{1,2,3,4})_*(\cE_0-\cE_2-\cE_4)=\cE_0-\cE_1-\cE_3,
$$
$$
(I_{1,2,3,4})_*(\cE_0-\cE_1-\cE_4)=\cE_0-\cE_2-\cE_3, \quad (I_{1,2,3,4})_*(\cE_0-\cE_2-\cE_3)=\cE_0-\cE_1-\cE_4.
$$
For any $i=5,\ldots,9$, the conic through $P_1,P_2,P_3,P_4,P_i$ is blown down to $P_i$. Since we are dealing with an involution, each $P_i$ is blown up to the conic through $P_1,P_2,P_3,P_4,P_i$. This gives us \eqref{I1234 5..9}.
Solving this system of linear equations results in the statement of proposition.
\end{proof}

Composing the linear maps given in \eqref{Ii 0}--\eqref{Ii i} and  in \eqref{I1234 0}--\eqref{I1234 5..9}, we finish the proof of the theorem. 
\end{proof}

\section{Involutions for pencils of biquadratic curves in $\mathbb{P}^1 \times \mathbb{P}^1$}
\label{s:2d_qrt}

\subsection{Horizontal and vertical switches}

We now discuss similar results for the symmetries of the QRT maps in the $\bbP^1\times\bbP^1$ formulation. Denote by $P_1,\ldots,P_8$ the base set of the pencil of biquadratic curves $\{\mathcal Q_\mu\}$, and let
\begin{equation}\label{Y}
Y=Bl_{P_1,\ldots,P_8}\bbP^1\times\bbP^1.
\end{equation}
This surface is birationally equivalent to \eqref{X}, and their Picard lattices are isomorphic:
\begin{equation}\label{Pic Y}
{\rm Pic}(Y)=\bbZ H_1\oplus\bbZ H_2\oplus\bbZ E_1\oplus\ldots\oplus\bbZ E_8.
\end{equation}
Here $H_1$, $H_2$  are the divisor classes of the total transforms of a generic vertical, resp. horizontal line, while $E_i$ are exceptional divisors of the blow-up of $\bbP^1\times\bbP^1$ at $P_i$. The Minkowski scalar product (intersection number) has the only non-vanishing products of the basis elements 
\begin{equation}\label{scal prod X}
\langle H_1,H_2\rangle=1,\quad \langle E_i,E_i\rangle =-1, \; i=1,\ldots 8.
\end{equation}
An isomorphism of the Picard lattices can be established via 
\[
\cE_0=H_1+H_2-E_1, \;\; \cE_1=H_1-E_1, \;\; \cE_2=H_2-E_1,\;\; \cE_{i+1}=E_i \;\;{\rm for}\; \; i=2,\ldots,8.
\]
The null root is now given by
\begin{equation}\label{delta Y}
\delta=2H_1+2H_2-E_1-\ldots-E_8.
\end{equation}
Some typical roots of the $E_8^{(1)}$ lattice are:
\begin{equation}
H_1-H_2,\quad E_i-E_j, \quad H_1-E_i-E_j,\quad H_2-E_i-E_j, \quad H_1+H_2-E_i-E_j-E_k-E_\ell. \quad
\end{equation}
A well-known result about the QRT map is the following.
\begin{proposition}\label{prop QRT}
The QRT map $f=i_1\circ i_2$ on $\bbP^1\times \bbP^1$  induces the action of $T_{H_1-H_2}$ on ${\rm Pic}(Y)$.
\end{proposition}
\begin{proof} One easily finds the maps induced on ${\rm Pic}(Y)$ by the vertical and the horizontal switches:
\begin{equation} \label{i1 H1 H2}
(i_1)_*(H_1)=H_1,\quad
(i_1)_*(H_2)=4H_1+H_2-\sum_{k=1}^8 E_k,
\end{equation}
\begin{equation}\label{i1 Ej}
(i_1)_*(E_j) = H_1-E_j, \;\; j=1,\ldots,8, 
\end{equation}
resp.
\begin{equation} \label{i2 H1 H2}
(i_2)_*(H_1) = H_1+4H_2-\sum_{k=1}^8 E_k,\quad
(i_2)_*(H_2) = H_2,
\end{equation}
\begin{equation}\label{i2 Ej}
(i_2)_*(E_j) = H_2-E_j, \;\; j=1,\ldots,8, 
\end{equation}
Now a simple computation shows that the action of $(i_1\circ i_2)_*$ on ${\rm Pic}(Y)$ coincides with what  the Kac formula \eqref{Kac formula} gives for the action of $T_{H_1-H_2}$:
\begin{eqnarray}
T_{H_1-H_2}(H_1) & = & H_1+2(H_1-H_2)+3\delta \;=\; 9H_1+4H_2-3\sum_{k=1}^8 E_k,\label{f H1}\\
T_{H_1-H_2}(H_2) & = & H_2+2(H_1-H_2)+\delta \;=\; 4H_1+H_2-\sum_{k=1}^8 E_k, \label{f H2}\\
T_{H_1-H_2}(E_j) & = & E_j+(H_1-H_2) +\delta\;=\; 3H_1+H_2-\sum_{k\neq j} E_k. \label{f Ej}
\end{eqnarray}
This finishes the proof.
\end{proof}

\subsection{Involutions along (1,1)-curves}

To the best of our knowledge, a geometric interpretation for translations corresponding to other roots of the $E_8^{(1)}$ lattice is not available in the literature. We are now going to close this gap.

\begin{definition}\label{def biquadratic inv 12}
Consider a pencil of biquadratic curves in $\bbP^1\times\bbP^1$ with base points $P_i$, $i=1,\ldots,8$. For any two distinct base points $P_i,P_j$, we define the involution $I_{i,j}$ as follows. Through a generic point $P$ passes a unique (2,2)-curve $Q_P$ of the pencil. Moreover, through $P_i,P_j,P$ there passes a unique (1,1)-curve $S$. We set  
$$
I_{i,j}(P) = \widetilde{P},
$$
where $\widetilde{P}$ is the fourth intersection point of $Q_P$ with $S$, so that $Q_P \cap S = \{P_i, P_j,P,\widetilde{P}\}$.
\end{definition}

These involutions have appeared previously in the literature: they are particular instances of more general involutions related to families of hyperelliptic curves of higher genus in $\bbP^1\times\bbP^1$, given in \cite[sect. 3.1]{NY}, while their $\bbP^2$ version was given in \cite[sect. 4.3]{PSWZ}. They will allow us to generate all the interesting translations in $W(E_8^{(1)})$.

\begin{theorem}\label{th biquadratic I12}\quad

\begin{enumerate}
\item The birational map $I_{i,j}\circ I_{i,k}$ on $\bbP^1\times\bbP^1$  induces the action of $T_{E_j-E_k}$ on ${\rm Pic}(Y)$.
\item The birational map $I_{i,j}\circ i_2$ on $\bbP^1\times\bbP^1$  induces the action of $T_{H_1-E_i-E_j}$ on ${\rm Pic}(Y)$.
\item The birational map $I_{i,j}\circ i_1$ on $\bbP^1\times\bbP^1$  induces the action of $T_{H_2-E_i-E_j}$ on ${\rm Pic}(Y)$.
\end{enumerate}
\end{theorem}

For the proof, one has to determine the action induced by $I_{i,j}$ on ${\rm Pic}(Y)$. 

\begin{proposition}\label{prop I12}
The action of $I_{i,j}$ on ${\rm Pic}(Y)$ is given by:
\begin{eqnarray}
(I_{i,j})_*(H_1) & = &3H_1+4H_2-3(E_i+E_j)-\sum_{k\neq i,j}E_k, \label{Iij H1}\\
(I_{i,j})_*(H_2) & = & 4H_1+3H_2-3(E_i+E_j)-\sum_{k\neq i,j}E_k, \label{Iij H2}\\
(I_{i,j})_*(E_i) & = & 3(H_1+H_2)-3E_i-2E_j-\sum_{k\neq i,j}E_k, \label{Iij Ei}\\
(I_{i,j})_*(E_j) & = & 3(H_1+H_2)-2E_i-3E_j-\sum_{k\neq i,j}E_k, \label{Iij Ej}\\
(I_{i,j})_*(E_k) & = & H_1+H_2-E_i-E_j-E_k, \;\; k\neq i,j. \label{Iij Ek}
\end{eqnarray}
\end{proposition}
\begin{proof}
We observe that a generic biquadratic curve of the pencil is mapped to itself, therefore
$$
(I_{i,j})_*(2H_1+2H_2-E_1-\ldots-E_8)=2H_1+2H_2-E_1-\ldots-E_8.
$$ 
Any (1,1)-curve through $P_i,P_j$ is mapped to itself, therefore
$$
(I_{i,j})_*(H_1+H_2-E_i-E_j)=H_1+H_2-E_i-E_j.
$$ 
The latter equation is actually a consequence of stronger relations. If $P$ belongs to the vertical line through $P_i$ then $\widetilde P$ belongs to the horizontal line through $P_j$, and vice versa, so that 
$$
(I_{i,j})_*(H_1-E_i)=H_2-E_j,  \quad (I_{i,j})_*(H_2-E_j)=H_1-E_i.
$$
Of course, we also can exchange here the roles of $i$ and $j$:
$$
(I_{i,j})_*(H_1-E_j)=H_2-E_i,  \quad (I_{i,j})_*(H_2-E_i)=H_1-E_j.
$$
For any $k\neq i,j$, the (1,1)-curve through $P_i,P_j,P_k$ is blown down to $P_k$. Since we are dealing with an involution, each $P_k$ is blown up to the (1,1)-curve through $P_i,P_j,P_k$. This gives us \eqref{Iij Ek}.
Solving this system of linear equations results in the statement of proposition.
\end{proof}

Now Theorem \ref{th biquadratic I12} follows by straightforward computations from \eqref{i1 H1 H2}--\eqref{i2 Ej}, Proposition \ref{prop I12}, and the Kac formula \eqref{Kac formula}. We omit the details.

\subsection{Involutions along (2,1)- and (1,2)-curves}

For biquadratic pencils, one can define a still another family of birational involutions.

\begin{definition}\label{def biquadratic inv 1234}
Consider a pencil of biquadratic curves in $\bbP^1\times\bbP^1$ with base points $P_i$, $i=1,\ldots,8$. For any four distinct base points $P_i,P_j,P_k,P_\ell$, we define the involutions $I^{(1)}_{i,j,k,\ell}$ and $I^{(2)}_{i,j,k,\ell}$ as follows. Through a generic point $P$ passes a unique biquadratic $Q_P$ of the pencil. Moreover, through $P_i,P_j,P_k,P_\ell,P$ there passes a unique (2,1)-curve $S^{(1)}$ and a unique (1,2)-curve $S^{(2)}$. We set  
$$
I^{(1)}_{i,j,k,\ell}(P) = \widetilde{P}^{(1)}, \quad I^{(2)}_{i,j,k,\ell}(P) = \widetilde{P}^{(2)},
$$
where $\widetilde{P}^{(1)}$ is the sixth intersection point of $Q_P$ with $S^{(1)}$, and $\widetilde{P}^{(2)}$ is the sixth intersection point of $Q_P$ with $S^{(2)}$.
\end{definition}

These involutions can be likewise used to give a geometric interpretation for translations in the group $W(E_8^{(1)})$.

\begin{theorem}\label{th biquadratic I1234}
The birational map $I^{(1)}_{i,j,k,\ell}\circ i_1=I^{(2)}_{i,j,k,\ell}\circ i_2$ on $\bbP^1\times\bbP^1$  induces the action of $$T_{H_1+H_2-E_i-E_j-E_k-E_\ell}$$ on ${\rm Pic}(Y)$.
\end{theorem}
The proof goes along the same lines as in the previous section.
\smallskip

As a side remark, we mention the following corollary of Theorems \ref{th biquadratic I12}, \ref{th biquadratic I1234}:
we have 
$$
I^{(1)}_{i,j,k,\ell}\circ i_1=I^{(2)}_{i,j,k,\ell}\circ i_2=I_{i,j}\circ i_1\circ I_{k,\ell}\circ i_2=I_{i,j}\circ i_2\circ I_{k,\ell}\circ i_1,
$$
therefore various involutions introduced above are related by
$$
I^{(1)}_{i,j,k,\ell}=I_{i,j}\circ i_2\circ I_{k,\ell}, \quad I^{(2)}_{i,j,k,\ell}=I_{i,j}\circ i_1\circ I_{k,\ell}.
$$

%
%
%

\section{Conclusions}
The results of the present work point out several directions in which they can be developed. These directions are under our current investigation.
\smallskip

1) In \cite{KMNOY, KMNOY1}, a geometric interpretation was given for discrete elliptic Painlev\'e equations corresponding to the translations $T_{\cE_i-\cE_j}$ in the Weyl group $W(E_8^{(1)})$. This interpretation is a non-autonomous version of the Manin map $I_j\circ I_i$, for the case when the nine points in $\bbP^2$ do not support a pencil of cubic curves. The resulting (non-autonomous) map corresponds to the base points moving in a certain way along one fixed cubic curve. It will be important to find a similar non-autonomous version of the geometric interpretation for discrete elliptic Painlev\'e equations corresponding to the translations $T_{\cE_0-\cE_i-\cE_j-\cE_k}$. 
\smallskip

2) Also a three-dimensional extension of the present results looks appealing. Recall that a birational realization of the Weyl group $W(E_7^{(1)})$ on $\bbP^3$, based on a configuration of eight points, was given in \cite{T}. Also for these maps, a geometric construction is not presently known. It can be anticipated that geometric involutions mentioned at the end of \cite{ASW1} and similar ones will allow us to find a geometric realization for all translations of this group, both in the autonomous and non-autonomous set-up.


\end{document}